\newcommand{\version}{September, 2017}
\documentclass[aps,pra,reprint,
amsmath,amssymb,superscriptaddress]{revtex4-1}
\usepackage{amsthm,amsmath,amsfonts,amssymb,amsxtra,appendix,bookmark,dsfont,latexsym,graphicx}

\makeatletter
\usepackage{hyperref}
\usepackage{delarray,color}
\usepackage{euscript}

\numberwithin{equation}{section}
\newcommand{\bdm}{\begin{displaymath}}
\newcommand{\edm}{\end{displaymath}}
\newcommand{\bdn}{\begin{eqnarray}}
\newcommand{\edn}{\end{eqnarray}}
\newcommand{\bay}{\begin{array}{c}}
\newcommand{\eay}{\end{array}}
\newcommand{\ben}{\begin{enumerate}}
\newcommand{\een}{\end{enumerate}}
\newcommand{\beq}{\begin{equation}}
\newcommand{\eeq}{\end{equation}}
\newcommand{\beqn}{\begin{eqnarray}}
\newcommand{\eeqn}{\end{eqnarray}}

\newcommand{\lf}{\left}
\newcommand{\ri}{\right}

\newcommand{\rv}{\mathbf{r}}

\newcommand{\diff}{\mathrm{d}}
\newcommand{\eps}{\varepsilon}

\newcommand{\glm}{\Psi^{\mathrm{GL}}}

\newcommand{\aav}{\mathbf{A}}

\newcommand{\hex}{b}
\newcommand{\theo}{\Theta_0}

\newcommand{\glfk}{\mathcal{G}_{\kappa,h_{\mathrm{ex}}}^{\mathrm{GL}}}

\newcommand{\glfe}{\mathcal{G}_{\eps}^{\mathrm{GL}}}
\newcommand{\gled}{e_{\eps}^{\mathrm{GL}}}
\newcommand{\glee}{E_{\eps}^{\mathrm{GL}}}

\newcommand{\curv}{k(s)}

\newcommand{\pot}{V_{k,\alpha}}

\newcommand{\curl}{\mbox{curl}}

\newcommand{\half}{\mbox{$\frac{1}{2}$}}

\newcommand{\tx}{\textstyle}

\newcommand{\R}{\mathbb{R}}

\newcommand{\E}{\mathcal{E}}

\newcommand{\OO}{\mathcal{O}}

\newcommand{\Om}{\Omega}

\newcommand{\dd}{\partial}

\newcommand{\Hc}{H_{\mathrm{c}1}}
\newcommand{\Hcc}{H_{\mathrm{c}2}}
\newcommand{\Hccc}{H_{\mathrm{c}3}}

\newenvironment{argument}{\noindent \emph{Argument.}}{\hfill \qed \medskip}
\newtheorem{teo}{Theorem}[section]
\newtheorem{statement}{Claim}[section]

\theoremstyle{remark}

\newcommand{\fO}{f_0}

\newcommand{\fone}{\E^{\mathrm{1D}}}

\newcommand{\fc}{\E ^{\rm{corr}}}

\newcommand{\eone}{E^{\mathrm{1D}}}

\newcommand{\eoneo}{E ^{\rm 1D}_{0}}

\newcommand{\alO}{\alpha_0}

\begin{document}

\title{Universal and shape dependent features of surface superconductivity}

\author{Michele Correggi}
\affiliation{Dipartimento di Matematica, ``Sapienza'' Universit\`{a} di Roma, P.le Aldo Moro, 5, 00185, Rome, Italy}

\author{Bharathiganesh Devanarayanan}
\affiliation{Department of Physics and Astronomy, National Institute of Technology Rourkela, 769008 Odisha, India}

\author{Nicolas Rougerie}
\affiliation{Universit\'e Grenoble-Alpes \& CNRS ,  LPMMC (UMR 5493), B.P. 166, F-38 042 Grenoble, France}

\date{\version}

\begin{abstract}
We analyze the response of a type II superconducting wire to an external magnetic field parallel to it in the framework of Ginzburg-Landau theory. We focus on the surface superconductivity regime of applied field between the second and third  critical values, where the superconducting state survives only close to the sample's boundary. Our first finding is that, in first approximation, the shape of the boundary plays no role in determining the density of superconducting electrons. A second order term is however isolated, directly proportional to the mean curvature of the boundary. This demonstrates that points of higher boundary curvature (counted inwards) attract superconducting electrons. 
\end{abstract}

\maketitle


The response of a superconductor to an applied magnetic field is well-known to be both very rich physically and very important practically ~\cite{deGennes-66,Tinkham-75,TilTil-90,Leggett-06}. In between the Meissner state, where the superconductor totally expels a weak magnetic field, and the normal state at very strong field, where superconductivity is destroyed, different types of mixed states may occur.

For a type-II superconductor in an external magnetic field, there are three critical values of the field marking phase transitions in the state of the material. In increasing order, one should distinguish the first critical field $\Hc$, where bulk superconductivity starts being challenged by the emergence of quantized vortices, from the second and third critical fields, $\Hcc$ and $\Hccc$. Between these latter values, bulk superconductivity is lost altogether, but Cooper pairs of superconducting electrons can still survive close to the boundary. It is a natural question to wonder how this surface superconductivity phenomenon depends on the shape of the sample. 

For a type-I superconductor however the transition from the superconducting to the normal state is more abrupt (first order transition) and only one critical value of the applied field is expected to be relevant. In fact, the behavior of a type-I superconductor can be much richer and show some connections with that of type-II superconductors \cite{Metal}. In spite of some recent progress \cite{Letal}, the understanding of superconductivity exactly at threshold between type-I and type-II materials, i.e., for $ \kappa = 1/\sqrt{2} $ is yet to be understood.

In this note we consider an infinitely long superconducting wire of smooth cross-section $\Omega \subset \R ^2$, modeled by a 2D Ginzburg-Landau (GL) theory set in $\Omega$, with a perpendicular external magnetic field (thus parallel to the axis of the wire). We consider values of the magnetic field between $\Hcc$ and $\Hccc$, so that all the physics happens close to the boundary (denoted by $\dd \Omega$) of $\Omega$. We present results in two complementary directions:

\medskip

\noindent$\bullet$ Universality: to leading order in the limit of large GL parameter, the physics does not depend on any detail of the surface's shape. The order parameter's amplitude is roughly constant in the direction tangential to $\dd \Omega$.

\smallskip

\noindent$\bullet$ Shape-dependence: the next-to-leading order of the energy density is directly proportional to the boundary's local curvature. Regions of larger boundary curvature (counted inwards) attract more superconducting electrons.

\medskip

Universality-type results date as far back as Saint James and de Gennes' seminal papers ~\cite{JamGen-63,deGennes-66}. They argued that, at the onset of surface superconductivity, the problem can be mapped to a one-dimensional one in the direction normal to the boundary. Essentially, we extend this view to the full surface superconductivity regime (where linearization of the model is \emph{not} valid). Such theoretical predictions (as well as more elaborated models of surface superconductivity) have been verified experimentally several times in the past, for different superconducting materials  \cite{CKKSSS05,MGSJ95,Ning-09,Retal03,RS67,SPSKC64,TGLS10,TLGFMA06}.

As for shape-dependent results, the influence of boundary curvature on the value of the third critical field has been known for some time: when increasing the magnetic field, superconductivity survives longer where the boundary's curvature is maximal. This has been demonstrated~\cite{SchPee-98,SchPee-99,JadRubSte-99,FomMisDevMos-98,FomDevMos-98} first in the case of domains with corners (i.e. points where boundary curvature jumps) that we do not discuss in detail here\footnote{It can be shown~\cite{CorGia-17a} that, as for regular samples, the energy as well as the the density of Cooper pairs is universal and unaffected by the presence of corners to leading order. The corrections due to corners are expected to show up in the next-to-leading order terms~\cite{CorGia-17b}.}. The mathematical analysis of the third critical field for domains with corners may be found in~\cite{BonDau-06,BonFou-07} (see~\cite[Chapter 15]{FouHel-10} for review and more references). As for smooth domains, the influence of boundary curvature was derived in~\cite{BerSte-98,HelMor-04} (see also~\cite[Chapters~8 and~13]{FouHel-10}). 

In this note, we summarize rigorous mathematical results obtained in~\cite{CorRou-14,CorRou-16,CorRou-16b}, discuss their physical interpretation, and complement them with numerical estimates of the shape-dependent contributions to surface superconductivity. The universality-type results we obtained had precursors in the mathematics literature~\cite{Almog-04,AlmHel-06,FouHelPer-11,FouHel-05,FouHel-06,FouKac-11,LuPan-99,Pan-02} (see~\cite[Chapter~14]{FouHel-10} for a review). The shape-dependent results however seem to be the first of their kind to be valid in the full regime of surface superconductivity. In view of experimental results~\cite{Ning-09} on the imaging of the surface superconductivity layer, it would be particularly interesting to measure the influence of the boundary's curvature on the concentration of Cooper pairs that we predict here. 

\section{Setting and model}

Let $\Omega \subset \R^2$ be a bounded 2D domain with smooth boundary. Our starting point is the GL energy functional for an infinite wire of cross-section $\Omega$, 
\begin{multline}\label{eq:gl func}
	\glfk[\Psi,\aav] = \int_{\Om} \diff \rv \: \bigg\{ \lf| \lf( \nabla + i h_{\rm ex}  \aav \ri) \Psi \ri|^2 - \kappa^2 |\Psi|^2 \\ + \half \kappa^2 |\Psi|^4  + \lf(h_{\rm ex}\ri)^2 \lf| \curl \, \aav - 1 \ri|^2 \bigg\}.
\end{multline}
Here $\kappa,h_{\mathrm{ex}}$ are the GL parameter and the strength of the applied magnetic field (parallel to the wire), respectively. The order parameter is denoted $\Psi$ and $h_{\mathrm{ex}} \aav$ is the vector potential of the induced magnetic field.   

It is well-known that surface superconductivity occurs in the regime $h_{\mathrm{ex}} \propto \kappa ^2$. To study this regime it will be convenient to change units, setting
\begin{equation}\label{eq:ext field} 
h_{\mathrm{ex}} = b \kappa ^2 = \eps ^{-2} 
\end{equation}
for two new parameters $b,\eps$. In these units we get
\begin{multline}\label{eq:GL func eps}
	\glfe[\Psi,\aav] = \int_{\Om} \bigg\{ \bigg| \bigg( -i\nabla + \frac{\aav}{\eps^2} \bigg) \Psi \bigg|^2 - \frac{1}{ \hex \eps^2} |\Psi|^2 \\ + \frac{1}{2 \hex \eps^2}|\Psi|^4  + \frac{1}{\eps^4} \lf| \curl \, \aav - 1 \ri|^2 \bigg\}.
\end{multline}
We consider the associated ground state problem 
\begin{equation}
	\glee : = \min_{(\Psi, \aav)} \glfe[\Psi,\aav],
\end{equation}
and study the limit $\eps \to 0$ with 
\begin{equation}\label{eq:b regime}
1 < b < \theo ^{-1} \approx 1.69
\end{equation}
where de Gennes' constant
\begin{equation}\label{eq:de Gennes}
\theo := \min_{\alpha \in \R } \min_{  \int |u| ^2 = 1 } \int_{0} ^{+ \infty} \diff t \lf\{ |\dd_t u | ^2 + (t+\alpha) ^2 |u| ^2 \ri\} 
\end{equation}
is the minimal ground state energy of the shifted harmonic oscillator on the half-line. In terms of the GL parameter $\kappa$ the limit $\eps \to 0$ corresponds to the extreme type II case $\kappa \to \infty$. In view of~\eqref{eq:ext field}  the condition~\eqref{eq:b regime} corresponds to asking that the external magnetic field lies strictly between $\Hcc$ and $\Hccc$. The numerical range for $ b $ in~\eqref{eq:b regime} is set by the lowest eigenvalues for Schr\"odinger operators with constant magnetic field equal to $1$, in the plane and the half-plane respectively. Note that $\eps$ is the only small parameter at our disposal (stressing this point is the main reason for changing units). Only when $b\to \theo ^{-1}$ (i.e. at the upper limit of the surface superconductivity regime) does the problem become linear. 

In all the sequel we take the above model as our starting point. Clearly this demands that the full BCS theory that would be appropriate reduces to the simpler GL model. The main condition is that one be sufficiently close to the critical temperature (see~\cite{Gorkov-59,deGennes-66} for standard references and~\cite{FraHaiSeiSol-12,FraHaiSeiSol-16} for a more recent and mathematical treatment). Reduction from a 3D to a 2D model is valid in two situations:

\smallskip

\noindent $\bullet$  a very long wire with magnetic field parallel to the wire;

\smallskip
 
\noindent $\bullet$ a very thin film with perpendicular magnetic field. 

\smallskip

In the latter case, one should directly impose that the induced magnetic field coincides with the external one, $\curl\, \aav = 1$ (see e.g.~\cite{AlaBroGal-10} and references therein), which only simplifies the following arguments.

\section{Shape independence of the Ginzburg-Landau energy}\label{sec:shape ind}

The first result we discuss is a two-term expansion for the GL ground state energy. All its parameters can be calculated by minimizing the reduced 1D energy functional   
\begin{equation}\label{eq:1D func bis}
\fone_{0,\alpha}[f] : = \int_0^{+\infty} \diff t \lf\{ \lf| \partial_t f \ri|^2 + (t + \alpha )^2 f^2 - \tx\frac{1}{2b} \lf(2 f^2 - f^4 \ri) \ri\}
\end{equation}
both with respect to the normal density profile $f:\R^+ \mapsto \R$ and the total phase circulation $\alpha \in \R$.  This is a generalization of the linear problem~\eqref{eq:de Gennes}, introduced in~\cite{Pan-02}. We denote the minimal value of the 1D effective energy by $\eoneo$ and by $\fO,\alO$ a minimizing pair. The above problem relates to the original one via scaling lengths by a factor of $\eps$ (typical thickness of the superconducting layer) in the direction normal to the boundary and assuming a particular ansatz for the GL minimizer, discussed below.  

Concerning the ground state energy, our results in~\cite{CorRou-14,CorRou-16,CorRou-16b} can be summarized as follows:

\begin{teo}[\textbf{GL energy expansion}]\label{theo:main}\mbox{}\\
For any $1<b<\theo ^{-1}$, in the limit $\eps \to 0$,
\begin{multline}\label{eq:main estimate}
\glee = - \frac{1}{2b} \bigg(C_1(b) \frac{|\dd \Om|}{\eps} 
\\ +  C_2 (b) \int_{\dd \Om} k(s) \diff s \bigg)  + o(1),
\end{multline}
where 
\begin{align}\label{eq:lead ord}
C_1(b) &= - 2 b \eoneo = \int_{0} ^{+\infty} \diff t \: \fO ^4 > 0 
\\
\label{eq:curv corr} C_2 (b) &=  \tx\frac{2}{3}b \fO ^2 (0) - 2b \alO \eoneo.
\end{align}
\end{teo}

\begin{proof}[Sketch of proof]
The methods of~\cite[Chapters~11 and~12]{FouHel-10} allow to prove that the order parameter and energy density are very small at distances larger than $\OO (\eps)$ from the boundary. Also, the induced magnetic field does not differ significantly from the applied one. Changing to scaled boundary coordinates $(s,t)=$(tangential, $\eps^{-1}\times$ normal), one is then lead to consider a scaled energy functional 
\begin{multline}\label{eq:intro GL func bound}
	\mathcal{G}[\psi] : = \int_0^{|\partial \Omega|} \diff s \int_0^{c_0 |\log\eps|} \diff t \lf(1 - \eps k(s) t \ri) \\ \lf\{ \lf| \partial_t \psi \ri|^2 + \frac{1}{(1- \eps \curv t)^2} \lf| \lf( \eps \partial_s - i t + \half i \eps k(s) t ^2 \ri) \psi \ri|^2 \ri.	\\	
	\lf. - \frac{1}{2 \hex} \lf[ 2|\psi|^2 - |\psi|^4 \ri]  \ri\},
\end{multline}
where $c_0$ is a fixed, essentially arbitrary, constant. The local curvature (counted inwards) $s\mapsto k(s)$ of the boundary  arises from the change of variables. The upshot is that, in first approximation (setting $\eps = 0$ in~\eqref{eq:intro GL func bound}), the curvature of the sample's boundary can be neglected. The whole sample is then mapped to a half-plane. To proceed further and get more precise results, one may instead approximate the boundary at each point by the auscultating circle of radius $k (s) ^{-1}$.

The particular case $|k| = R^{-1} = \mathrm{const.}$ for the above functional corresponds to the surface energy in a disk of radius $R$ (or the exterior thereof if $k<0$). Since all terms involving the $s$-coordinate in~\eqref{eq:intro GL func bound} come with $\eps$ pre-factors, it makes sense to guess an ansatz of the form $f(t) e ^{-i \frac{\alpha}{\eps} s}$ with a constant $\alpha$ to solve the disk problem. This leads to considering the reduced functional
\begin{multline}
\label{eq:1D func}
\fone_{k,\alpha}[f] : = \int_0^{c_0|\log\eps|} \diff t (1-\eps k t )\big\{ \lf| \partial_t f \ri|^2 + \pot(t) f^2 \\ - \tx\frac{1}{2b} \lf(2 f^2 - f^4 \ri) \big\},
\end{multline}
where
\begin{equation}
	\label{eq:pot}
	\pot(t) : = \frac{(t + \alpha - \frac12 \eps k t ^2 )^2}{(1-\eps k t ) ^2}.
\end{equation}
This is a generalization of~\eqref{eq:1D func bis} to the case of non-zero curvature. The main step of the proof is to vindicate the exactness of the ansatz $\psi (s,t) = f(t) e ^{-i \frac{\alpha}{\eps} s}$ for the disk problem ($k = \mathrm{const.}$) in~\eqref{eq:intro GL func bound}. 

For simplicity we explain this in the case $k=0$, corresponding to the degenerate case of a half-plane sample. The method is to write a tentative minimizer in the form 
$$\psi (s,t) = \fO (t) e^{-i \frac{\alO s}{\eps}} v(s,t)$$
and bound from below the excess energy due to a possibly non constant $v$ in the manner
$$
 \E_0 [v] \geq \int_0 ^{|\dd \Omega|} \diff s \int_{0} ^{+\infty} \diff t \:  \left(\fO ^2(t)  + F_0 (t) \right) \lf| \nabla v \ri|^2,
$$
with a potential function 
$$ F_0 (t) = 2 \int_{0} ^{t} \diff \eta \: (\eta + \alO)\fO ^2 (\eta).$$
The ``cost function'' $\fO ^2(t)  + F_0 (t) $ can be interpreted as a lower bound to the kinetic energy density generated by a vortex located at distance $t$ from the boundary. In~\cite{CorRou-14} we prove that $\fO ^2(t)  + F_0 (t) \geq 0$ for any $t$, provided $1\leq b \leq \theo^{-1}$, thus concluding that a vortex, or more generally any non constant $v$, would increase the energy. The ground state energy of the half-plane problem is thus simply given by $|\dd \Omega| \eoneo$. The disk case follows similar considerations, with significant but technical additional difficulties.

Relying on the previous observations, we get a kind of ``adiabatic'' decoupling where the order parameter is for any $s$ in the ground state of the problem in the direction perpendicular to the boundary. This reduces matters to looking for the energy $\eone_\star \left(k\right)$ obtained by minimizing~\eqref{eq:1D func} with respect to both $\alpha\in \R$ and $f:\R^+ \mapsto \R$. The rationale is that a true GL minimizer can, in a suitable gauge, be approximated in the manner 
\begin{multline}\label{eq:intro GLm formal refined}
\glm(\rv) = \glm (s,\tau) \approx f_{k(s)} \left(\tx \frac{\tau}{\eps} \right)  \exp \left( - i \alpha (k(s)) \tx \frac{s}{\eps}\right) \\ \exp\lf( i \phi_{\eps}(s,\tau) \ri) 
\end{multline}
with $f_{k(s)},\alpha(k(s))$ a minimizing pair for the energy functional~\eqref{eq:1D func} at curvature $k = k(s)$ and $(s,\tau) = $ (tangential coordinate, unscaled normal coordinate). Here $\phi_\eps$ is a phase factor accounting for a uniform vortex density proportional to $\eps ^{-2}$ in the bulk of the sample, which is required to compensate the huge magnetic field inside the sample. Such a large circulation gets corrected by a relatively small amount $- \alpha(k(s)) \eps ^{-1}$ in order to optimize the boundary energy. This form includes subleading order corrections, and its main feature is the decoupling of scales/coordinates.

The final expression~\eqref{eq:main estimate} follows from first order perturbation theory applied to the 1D functional~\eqref{eq:1D func}. Since we work in the regime $\eps \to 0$, one can expand $\eone_\star \left(k(s)\right)$ in powers of $\eps$. The leading order is given by the $k=0$ functional and the first correction by $- \eps k  $ times 
\begin{multline}\label{eq:corr func}
 \fc _{\alpha_0} [f_0] :=  \int_{0} ^{c_0 |\log \eps|} \diff t \:  t\bigg\{ \lf| \partial_t f_0 \ri|^2 \\ 
+ f_0^2 \left( -\alpha_0 (t+\alpha_0) -\frac{1}{b} + \frac{1}{2b} f_0 ^2\right)\bigg\}.
\end{multline}
That $C_1 (b)$ is equal to the expression in the right-hand side of~\eqref{eq:lead ord} is a simple consequence of the variational equation satisfied by $\fO$, while showing that \eqref{eq:corr func} coincides with \eqref{eq:curv corr} divided by $ 2b $ requires a few non trivial computations that we skip for brevity. 
\end{proof}

For $b \to \theo^{-1}$, the half-plane problem can be linearized, and the proof is much simpler. Indeed, for the linear problem, the ansatz $\fO (t) e^{-i \frac{\alO}{\eps} s}$ is clearly optimal (as used first in~\cite{JamGen-63} and further justified in the mathematics literature \cite[Chapter~13]{FouHel-10}). This is because the Hamiltonian of the linear problem commutes with translations in the tangential coordinate, so that one can perform a Fourier series decomposition in this direction. This nice structure is broken by the nonlinear term. In this respect, one may see Theorem~\ref{theo:main} as an extension  to the nonlinear regime $b<\theo ^{-1}$ of the linear analysis pioneered in~\cite{JamGen-63}. 

The energy asymptotics~\eqref{eq:intro GLm formal refined} also implies estimates of the amplitude of the order parameter and the phase circulation around the boundary. Indeed, $|\glm|^2$ can be shown to be close to $|\fO \left(\tx \frac{\tau}{\eps} \right)| ^2$ \emph{pointwise} in the surface superconductivity layer. This indicates that no defect can be present therein, in particular no vortex. 

As we have anticipated, the curvature effects emerge in the next-to-leading order corrections to the density of superconductivity. However, in the energy expansion~\eqref{eq:main estimate}, the second term is universal and independent of the sample's shape: by the Gauss-Bonnet theorem, we have that 
$$
\glee = \frac{|\dd \Om|\eoneo}{\eps} -\frac{\pi}{b}  C_2 (b) + o(1),
$$
for the integral of the curvature in~\eqref{eq:main estimate} is just $ 2\pi $ times the Euler characteristic of the domain $ \Omega$, equal to $1$. In the next section we turn to describing the most explicit form of curvature dependence we are able to derive.

\section{Shape dependence in the distribution of superconductivity}

The main result we obtained in~\cite{CorRou-16b} is stated in terms of the distribution of $|\glm| ^4$, the square of the Cooper pairs' (normalized) density. For obvious physical reasons it would be desirable to obtain instead an estimate of $|\glm| ^2$, but this does not follow straightforwardly from our methods. 

\begin{teo}[\textbf{Curvature dependence of $\glm$}]\label{theo:main bis}\mbox{}\\
Let $\glm$ be a GL minimizer and $D\subset \Omega$ be a measurable set intersecting $\dd \Omega$ with right angles. For any $1<b<\theo ^{-1}$, in the limit $\eps \to 0$,
\begin{multline}\label{eq:main estimate curv}
\int_{D} \diff \rv \: |\glm| ^4 = \eps \, C_1(b) |\dd \Om \cap \dd D| \\ 
+ \eps ^2 C_2 (b) \int_{\dd D\cap \dd \Om} k(s) \diff s + o(\eps ^2),
\end{multline}
using the notation of Theorem~\ref{theo:main}.
\end{teo}

\begin{proof}[Sketch of proof]
Theorem~\ref{theo:main} is actually proved in a local fashion. Denoting $\gled(\rv)$
the Ginzburg-Landau energy density of a ground state, we in fact obtain, for any domain $D\subset \Omega$, a version of~\eqref{eq:main estimate} for the energy density integrated over $D$. The main idea is then that, neglecting the magnetic kinetic energy (last term in~\eqref{eq:GL func eps}), we essentially have 
\begin{equation}\label{eq:ener dens app}
\gled \approx -\frac{1}{2b\eps ^2} |\glm| ^4
\end{equation}
and thus the result. The above (approximate) identity follows by multiplying the GL equation by $\glm$ and integrating the result over a domain $D\subset \Omega$. Boundary terms however arise due to integrating the kinetic energy by parts, involving the derivative of $\glm$ in the direction normal to $\dd D$. They are negligible for domains intersecting $\dd \Omega$ with $\pi/2$ angles because, for such domains, this direction is tangential to the boundary $\dd \Omega$ of the sample. The variations of $|\glm|$ in this direction are very slow, as can be guessed from~\eqref{eq:intro GLm formal refined}. 
\end{proof}

Theorem~\ref{theo:main bis} is almost as far as we can get with an analytic approach. It is however not quite obvious to determine the sign of $C_2(b)$ from the expression we obtained, i.e. to conclude whether larger curvature favors superconductivity, or the other way around. The expression~\eqref{eq:curv corr} can however be evaluated by numerically solving a rather simple variational problem. We next report on the conclusions of such an investigation:

\begin{statement}[\textbf{Curvature favors superconductivity}]\label{sta:numerics}
For any $1<b<\theo ^{-1}$, we find that $C_2 (b) > 0$. Consequently, the distribution of surface superconductivity along the boundary is an increasing function of the local curvature (counted inwards).
\end{statement}

\begin{figure}\label{fig:ener corr}
\includegraphics[width=9cm]{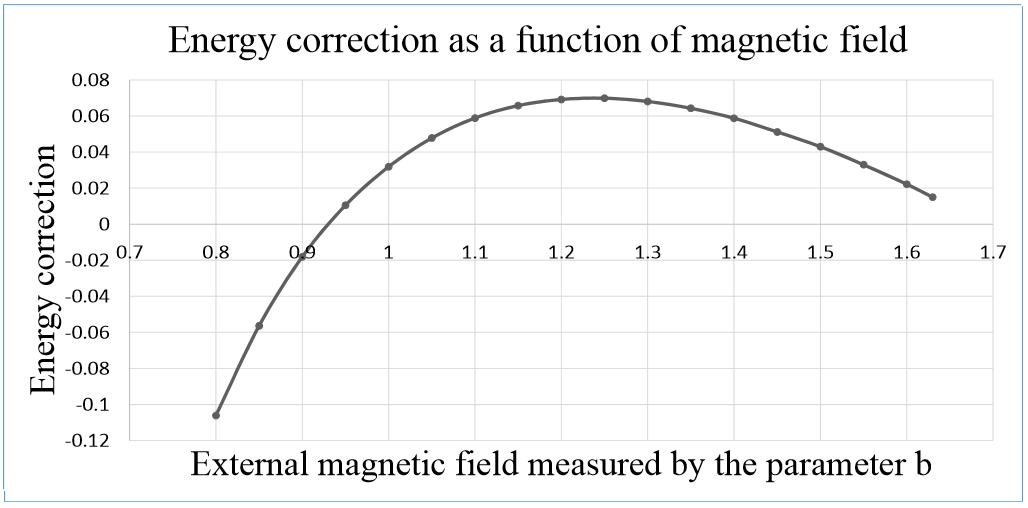}
\caption{Energy correction $C_2 (b)$ from \eqref{eq:lead ord} as a function of magnetic field, measured by the parameter $b$, cf~\eqref{eq:ext field}.}
\end{figure}

\begin{argument}
Everything boils down to solving the 1D minimization problem~\eqref{eq:1D func bis} and inserting the results in~\eqref{eq:curv corr}. We proceed as follows:

\medskip

\noindent $\bullet$ Starting from an initial value for $\alpha$, we minimize $\fone_{0,\alpha}[f] $ as a function of the normal profile $f$. We find it simpler to consider a problem on the full real line, which is equivalent and avoids boundary conditions issues. We thus reflect the harmonic potential and minimize
\begin{multline*}
\int_{-\infty} ^{+\infty} \diff t \big\{ \lf| \partial_t f \ri|^2 + \min\left((t + \alpha )^2,(t-\alpha) ^2 \right) f^2 \\ 
- \tx\frac{1}{2b} \lf(2 f^2 - f^4 \ri) \big\}.
\end{multline*}
Up to the fact that $f$ should not be normalized, this can be seen as the computation of the ground state profile of a Bose-Einstein condensate in the double-well potential $\min\left((t + \alpha )^2,(t-\alpha) ^2 \right) - b ^{-1}$. We rely on the freely available Matlab toolbox GPELab~\cite{GPELab,AntDub-14,AntDub-15}, removing the normalization step imposed therein.

\smallskip

\noindent $\bullet$ Having computed the minimizing profile $f_{0,\alpha}$ at fixed $\alpha$, the minimization in $\alpha$ at fixed $f$ is very easy: the Feynman-Hellmann principle yields 
$$ \int_{0} ^\infty (t + \alpha) f_{0,\alpha} ^2 (t) \diff t = 0$$
and we use this to update the value of $\alpha$.

\smallskip

\noindent $\bullet$ We again minimize the energy with respect to $f$, using the updated value of $\alpha$, and iteratively continue the procedure until convergence is reached.

\smallskip

\noindent $\bullet$ A few benchmarks can be used to ascertain the validity of the calculation. The second identity in~\eqref{eq:lead ord} should always hold for the converged quantities. Next, the two expressions~\eqref{eq:curv corr} and $C_2 (b) = 2 b \fc_{\alO} [f_0]$
should always coincide. Finally, we checked that the cost function $\fO ^2(t)  + F_0 (t) $ is always positive for $1 < b < \theo ^{-1}$.

\medskip

For $b$ very close to $\theo^{-1}$, the numerics become rather unstable, for then $\fO$ takes very small values. This regime would be better analyzed by solving the lowest eigenvalue problem for the linearized equation. We have not pursued this beyond the available studies (e.g.~\cite{Bonnaillie-08}), for the result we aim at in this regime is already included in~\cite[Lemma~2.3]{CorRou-16b}. Here, perturbation theory in $|b-\theo ^{-1}|$ proves that there exists $\delta >0$ such that $C_2 (b) > 0$ for any $b$ satisfying $\theo ^{-1} - \delta \leq b \leq \theo ^{-1}$. This does not settle the question analytically for the full parameter regime of interest but allows to bypass the numerical instability issue for $b \to \theo ^{-1}.$

The values of the energy corrections $C_2 (b)$ we found are plotted in Figure~1. The result is clearly positive for $1<b<\theo ^{-1}$, as claimed in the statement. The numerical problem can also be tackled for $b<1$, but this is less relevant physically, for then bulk superconductivity starts challenging surface superconductivity, which is not taken into account in the reduced model we solve. 
\end{argument}

\section{Remarks on the density profile}

As an illustration, we provide plots of the optimal density profile at $b=1$ and $b=1.5$ in Figure~2. According to the  discussion in Section~\ref{sec:shape ind}, they give a very good approximation for the dependence of the (amplitude of the) order parameter in the direction normal to the boundary. In this respect, it is interesting to note that the density is not monotonously decreasing as a function of the distance to the boundary. This can be seen analytically because
$$ f_0' (0) = 0, \quad f_0'' (0) = f_0(0) \left(b ^{-1} - \alO ^2\right) \geq 0$$
where the first equation is the Neumann boundary condition and the second uses the variational equation for $f_0$ (Euler-Lagrange equation associated to~\eqref{eq:1D func bis}) together with the known fact~\cite[Proof of Lemma~3.3]{CorRou-14} that $f_0^2 (0) = 2 - 2 b \alO ^2 \geq 0$. Thus, $f_0$ increases close to the origin, whereas it must ultimately decay, so that it has to reach at least one global maximum outside the origin. Since superconductivity is due to boundary effects in this regime, it is rather counter-intuitive that the maximum density of Cooper pairs is not attained exactly at the boundary.

In our simulations, this drop of density at the boundary is hardly visible for values of $b$ close to the third critical field (see the curve for $b=1.5$ in Figure~2) but becomes more pronounced for $b$ close to $1$, i.e. close to the second critical field (as per~\eqref{eq:ext field}). It might thus be a precursor of the emergence of bulk superconductivity in the sample, when the second critical field is crossed from above.  

Note that the maximum density not occuring at the boundary is a nonlinear effect: the linearized model valid close to $\Hccc$ leads to a monotonously decreasing density. Indeed, close to $\Hccc$ the density is found by minimizing~\eqref{eq:1D func bis} without nonlinear term, under a mass unit mass constraint~\cite{JamGen-63,FouHel-10}. This leads to the variational equation 
$$ - f'' + (t+\alpha)^2 f= \lambda f$$
and it is known~\cite[Chapter~3]{FouHel-10} that the value of $\alpha$ giving the minimal energy is such that $\alpha  = - \sqrt{\lambda} = -\sqrt{\theo}$, so that one gets 
$$ f''(t) = tf(t) \left(t - 2 \sqrt{\theo}\right). $$
Since the minimizing $f$ must be positive, the above changes sign only once on the positive half-axis, and this is incompatible with a global maximum occurring away from the origin. 

Since for $b\to \theo^{-1}$ the minimizer of~\eqref{eq:1D func bis} gets closer and closer to a multiple of the linear solution, it is natural that the density maximum away from the origin is less visible in that limit.

\begin{figure}\label{fig:dens}
\includegraphics[width=7cm]{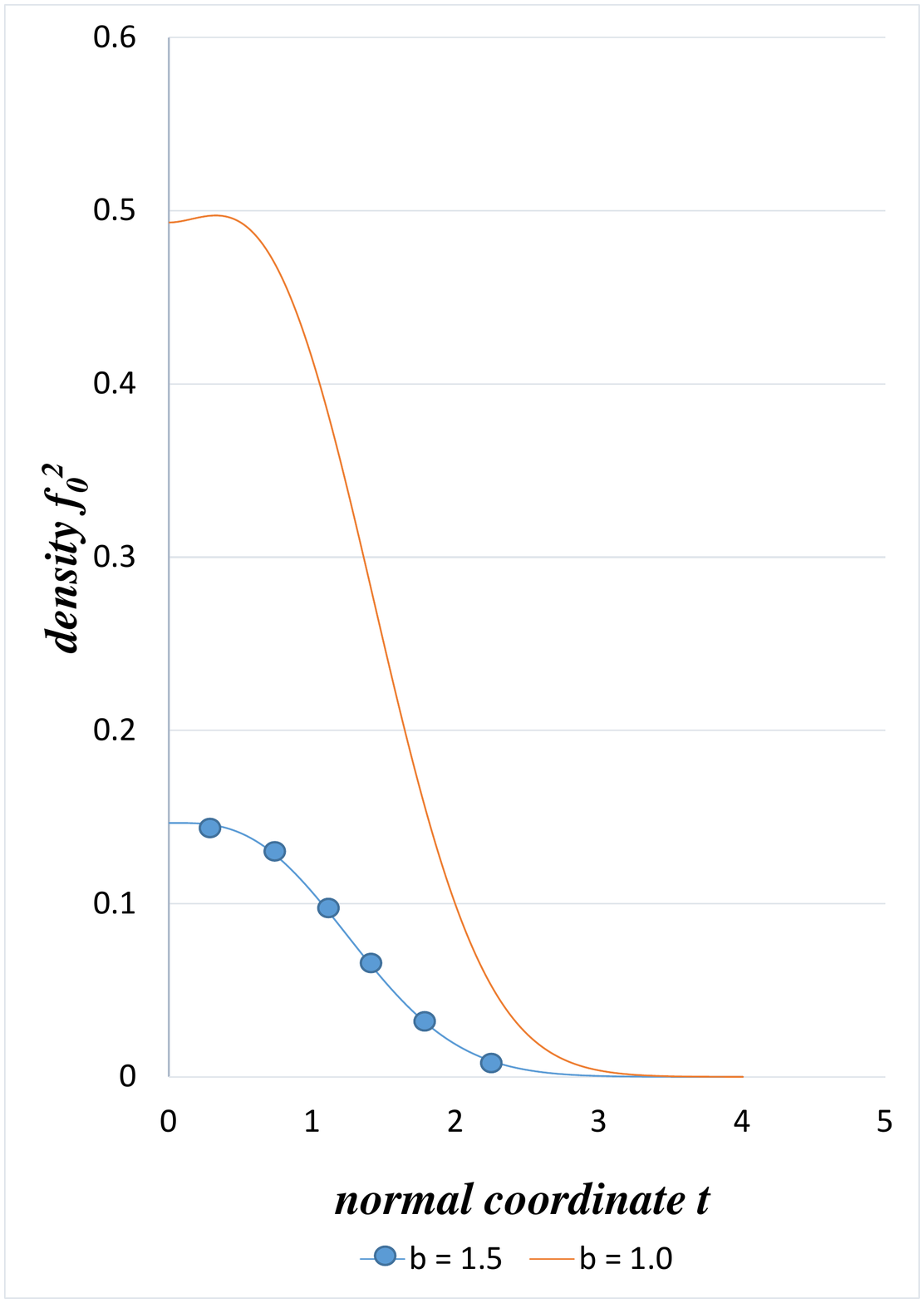}
\caption{Optimal density profile $\fO ^2 (t)$ for the 1D reduced problem~\eqref{eq:1D func bis} as a function of the scaled normal coordinate, counted inwards. Orange solid curve for $b=1$ and blue solid curve with dots $b=1.5$. The corresponding values of $\alO$ are respectively 0.153 and 0.786.}
\end{figure}
%
%

\section{Conclusions}

We studied the surface superconducting regime of the usual two-dimensional GL theory for a very elongated cylinder in a magnetic field perpendicular to its cross-section (also valid for a thin film with perpendicular field). Our results hold in the limit of a large GL parameter, in the full regime of applied field between the second and third critical values $\Hcc$ and $\Hccc$. In particular, we showed that:

\smallskip

\noindent\textbf{1.} To leading order, the distribution of superconducting electrons along the sample's boundary is independent of its shape. The boundary can in first approximation be treated as a single straight line, and the density of Cooper pairs has significant variations only in the direction perpendicular to it. 

\smallskip

\noindent\textbf{2.}  The subleading correction to the distribution is directly proportional to the local value of the boundary's curvature. A larger curvature (counted inwards) corresponds to a larger density of Cooper pairs. This follows by locally approximating the boundary by an auscultating circle.

\smallskip

These results are in accord with what was previously obtained from a linearized analysis close to the upper critical field $\Hccc$ (see~\cite{FouHel-10} for review). They however constitute, as far as we know, the first complete results of their kind valid for all the relevant values of the applied field, where a fully nonlinear analysis is required. 

We remark that our approach requires the boundary curvature to be smooth, i.e. that its variations happen on a scale much larger than the coherence length (proportional to $\kappa^{-1} \propto\eps$). If this assumption is violated at some points of the boundary, it would be appropriate to treat these as corners (where the curvature makes sharp jumps). Such an analysis will be provided elsewhere~\cite{CorGia-17a,CorGia-17b}. 

As directions for future investigations we mention the case where the magnetic field is not strictly perpendicular to the cross-section (or the film), and a more complete 3D treatment in case the wire is bent for example. In the latter case, it is known that superconductivity appears first, when decreasing the external field below $\Hccc$, in regions where the magnetic field is tangent to the boundary~\cite{FouHel-10}. In the regime of our interest, superconductivity should however be present along the whole surface. The first question to tackle then is certainly the influence of the angle between the magnetic field and the cross-section. Results in this direction may be found in~\cite{FouKacPer-13}. In a full 3D treatment, whether the local (Gaussian) curvature of the boundary rather than the 2D cross-sectional curvature influences the Cooper pairs' distribution remains an open question.

\bigskip

\noindent\textbf{Acknowledgments:} M.C. acknowledges the support of MIUR through the FIR grant 2013 ``Condensed Matter in Mathematical Physics (Cond-Math)'' (code RBFR13WAET). N.R. acknowledges the support of the ANR project Mathostaq (ANR-13-JS01-0005-01). B.D. was supported by the CNRS Inphynity project MaBoLo. We thank Romain Duboscq for helpful advice regarding the GPELab toolbox. 


\end{document}